\documentclass[runningheads]{llncs}

\usepackage[margin=1in]{geometry}
\usepackage{xcolor}
\usepackage[numbers,sort&compress]{natbib}

\usepackage{url}
\usepackage{amsmath}
\usepackage[shortlabels]{enumitem}
\usepackage{harpoon}
\usepackage{etoolbox}
\usepackage[linesnumbered]{algorithm2e}
\usepackage{float}
\usepackage{xspace}
\usepackage{tikz}
\usetikzlibrary{fit,shapes.geometric,arrows,automata,arrows.meta}
\usepackage{mathrsfs}
\usepackage{amsfonts}
\usepackage{amssymb}
\usepackage{hyperref}
\usepackage[nameinlink,capitalise]{cleveref} 
\usepackage[misc,geometry]{ifsym}
\def\bb{\mathbb}
\def\cal{\mathcal}
\def\sube{\subseteq}

\def\Z{{\bb Z}}

\def\I{{\cal{I}}}
\def\<{\langle}
\def\>{\rangle}
\def\rank{\operatorname{rank}}

\def\level{level}
\def\e{\epsilon}
\def\conv{\operatorname{conv}}
\def\OPT{\text{OPT}}
\def\r{k}

\def\[#1\]{\begin{align*}#1\end{align*}}
\long\def\addlprow#1{& #1 \vspace{3pt}\\}
\long\def\linearprogram#1#2#3#4{
    \tag*{{#1}}
    \begin{array}{r@{\hspace{10pt}}l@{\hspace{20pt}}l}
        #2 & \multicolumn{2}{@{}l}{#3} \vspace{10pt} \\
        \text{such that}
        \forcsvlist\addlprow{#4}
    \end{array}
}

\spnewtheorem{observe}{Observation}{\bfseries}{\rmfamily}

\title{A Knapsack Intersection Hierarchy Applied to All-or-Nothing Flow in Trees}
\author{Adam Jozefiak\inst{1} \and
F. Bruce Shepherd\inst{1} \and
Noah Weninger\inst{2}}
\authorrunning{Adam Jozefiak, F. Bruce Shepherd, and Noah Weninger}
\institute{University of British Columbia, Vancouver, Canada\\ \url{adam.jozefiak@alumni.ubc.ca}, \url{fbrucesh@cs.ubc.ca} \and
University of Waterloo, Waterloo, Canada\\\url{nweninge@uwaterloo.ca}}

\begin{document}

\maketitle

\begin{abstract}
    We introduce a natural {\em knapsack intersection hierarchy} for strengthening
    linear programming relaxations of
    packing integer programs, i.e., $\max\{w^Tx:x\in P\cap\{0,1\}^n\}$ where $P=\{x\in[0,1]^n:Ax \leq b\}$ and $A,b,w\ge0$.
    The $t^{th}$ level $P^{t}$ corresponds to adding cuts associated with the integer hull
    of the intersection of any $t$ {\em knapsack constraints} (rows of the
    constraint matrix). This model captures the maximum possible strength of
    ``$t$-row cuts'', an approach often used by solvers for small $t$. If
    $A$ is $m \times n$, then $P^m$ is the integer hull of $P$ and $P^1$
    corresponds to adding cuts for each associated single-row knapsack problem. Thus, even
    separating over  $P^1$ is NP-hard.  However, for fixed $t$ and any
    $\epsilon>0$, results of Pritchard imply there is a polytime
    $(1+\epsilon)$-approximation for $P^{t}$.  We then
    investigate the hierarchy's strength in the context of the well-studied
    all-or-nothing flow problem in trees (also called unsplittable flow
    on trees). For this problem,  we show that the integrality gap of $P^t$ is $O(n/t)$ and give examples where  the gap is $\Omega(n/t)$.
   We then examine  the stronger formulation $P_{\rank}$ where all rank constraints are added.
 For $P_{\rank}^t$, our best
    lower bound drops to $\Omega(1/c)$ at level $t=n^c$ for any $c>0$.
    Moreover, on a well-known class of ``bad instances'' due to Friggstad and Gao, we show that we can achieve this gap;  hence  a constant integrality gap for these instances is obtained at level $n^c$.
\end{abstract}

\section{Introduction}
\label{sec:knapsack_hierarchy}

In this paper we study linear relaxations for {\em packing integer programs} (PIP).
A PIP is described by a 0-1 optimization problem
$\max\,\{w^Tx:x\in\{0,1\}^n,\,Ax\le b\}$, where $w\in\Z^n_+,\,A\in\Z^{m\times n}_+,\text{ and }b\in\Z^m_+$. These integer programs capture well-known problems such as 0-1 knapsack, matroid optimization,
maximum stable set, demand matching and
all-or-nothing flow in trees (also called unsplittable flow on trees).
PIPs are also called 0-1 multidimensional knapsack
problems in the case where $m$ is fixed.
We introduce a hierarchy of strengthened PIP formulations
where level $t$ is defined by adding cuts associated with
the integer hulls of all intersections of $t$ constraints.
This {\em knapsack intersection hierarchy} is inspired by successful computational approaches;
in the case of a single constraint it corresponds to the cuts added in the
pioneering work of Crowder, Johnson, and Padberg \cite{crowder1983solving}.
We evaluate the strength of this hierarchy applied to
 the well-studied ``all-or-nothing flow''  problem in trees (ANF-Tree). This  problem generalizes weighted matching but has no known polytime $O(1)$-approximation.
 In this section we formally define the hierarchy and in the next section we discuss our results for ANF-Tree.

PIPs generally have a (one or more) natural linear relaxation $P:=\{x\in[0,1]^n:Ax\le b\}$, where $A,b \geq 0$. The discrete problem of interest is to optimize over
 the {\em integer hull}  $P_I:=\conv(P\cap\{0,1\}^n)$.
Computational solution strategies for PIPs often use some form of branch and cut
method, and one of the most effective approaches is to rely on cuts for the
knapsack polytopes associated with individual constraints
\cite{crowder1983solving}. Let $a^j$ be row $j$ of $A$ and $b^j$ be element $j$ in $b$.
For each $j\in[m]$, let $K(j)$ denote the
polytope $\{x \in [0,1]^n:  \sum_ia^j_i x_i \leq b^j\}$.
The {\em knapsack cuts} for $K(j)$ are the inequalities
which are valid for the integer hull $K_I(j):=\conv(K(j)\cap\{0,1\}^n)$, i.e., the knapsack polytope for constraint $j$. On each iteration of a
branch and cut approach (e.g., see \cite{schrijver1998theory}), one has a
feasible---but fractional---solution $\tilde{x}$ for a current relaxation $P'$ of
$P_I$. In \cite{crowder1983solving},  they generate knapsack cuts for some constraint. That is,
for some $j$, they find a valid inequality $c^Tx \leq d$ for $K_I(j)$ for which
$c^T \tilde{x} > d$. Adding such inequalities to $P'$ gives a tighter formulation for
$P_I$ on which to recurse.

This approach has also been extended to {\em multi-row cuts}. This can be set up in various ways, e.g.: (1)
by aggregating multiple constraints to form a single inequality and then generating cuts for the associated knapsack \cite{dey2014practical,xavier2017computing,dey2018theoretical} and (2) by
considering cuts associated with the integer hull of the intersection of
several knapsack polytopes
\cite{louveaux2008polyhedral,kellerer2004multidimensional}. The latter set-up
is potentially stronger in the following sense: there are instances  where adding all
cuts of type (2) defines the integer hull but adding (any number of)
cuts of type (1) does not.\footnote{A well-known example for the Chv\'atal rank
actually shows that one may need an unbounded number of rounds of aggregated
cuts in order to obtain the integer hull (e.g., see Section 23 in \cite{schrijver1998theory}).}

We discuss a framework to measure the strength of cuts in the latter setting.
For some $S \subseteq [m]$, we denote the intersection of the fractional knapsack polytopes associated with constraints in $S$; we define
$K(S) := \cap_{j \in S} K(j)$.
We consider a relaxation where all  cuts are added for the associated integer hull $K_I(S)$.
We then define a {\em knapsack intersection hierarchy}  for
$P_I$ as follows.  For each $t\in[m]$,  define
\begin{align*}
    P^t := \bigcap_{|S|=t}K_I(S).
\end{align*}

\noindent In other words, $P^t$ is  obtained from $P$ by adding,
for each $S \subseteq [m]$ with $|S| = t$, all valid inequalities for $K_I(S)$.
Clearly, $P^{t+1} \subseteq P^t$ and $P^m=P_I$, so we have the
following hierarchy:
\[
P \supseteq P^1 \supseteq \ldots \supseteq P^{m-1} \supseteq P^m = P_I.
\]

\noindent Separating over $K_I(S)$ is NP-Hard given that, even for $t=1$,  0-1 knapsack is a special case,
Hence it is already NP-Hard to separate over $P^1$, the first level of the hierarchy; this fate is shared by a different hierarchy, since the Chv\'atal closure
of a polyhedron is NP-hard to separate \cite{eisenbrand1999note}.
To mitigate this, we show that results of Pritchard \cite{pritchard2010lp} lead to a tractable formulation, that is, one that is polynomially sized, but
approximate,  when $t$ is constant.

\begin{theorem}
    \label{prop:approx}
    For $0<\epsilon\le 1$, there is an approximate formulation for $P^t$ of size $O(n^{t^3\epsilon^{-1}+t+1})$ for which the value of an optimal solution is at most
    a $(1/(1-\epsilon))$-factor larger than the optimal solution to $P^t$.
\end{theorem}

\begin{corollary}
    \label{cor:approx}
    For fixed $t$ there is a PTAS for $\max\{w^Tx : x\in P^t\}$.
\end{corollary}

We defer the proofs to \cref{sec:first2proofs}. We now discuss the impact of the knapsack hierarchy on formulations for ANF-Tree.

\subsection{All-or-Nothing Flow in Trees}

The {\em all-or-nothing flow problem} \cite{chekuri2007multicommodity} is defined for a multiflow problem whose input is a supply graph $G$ and demand graph $H$.  $G$ and $H$ may also be endowed with edge capacities $u_e: e \in E(G)$ and demands $d(f): f \in E(H)$. We call $E(H)$ the requests and  a subset $R$ of requests  is {\em routable} if there is a (fractional) multiflow which routes the requests in $R$ using $G$'s capacity. The problem is ``all-or-nothing'' in the sense that if $f \in R$, then we must route the whole $d(f)$ units of demand. An instance is said to satisfy the
{\em no-bottleneck-assumption} (NBA) if $d(f) \leq u_e$ for every request $f$ and supply edge $e$.
ANF-Tree is the special case of all-or-nothing flow where the supply graph $G$ is a tree.

When the NBA holds,
there is a polylog approximation in general graphs
\cite{chekuri2004all} and a
 $48$-approximation in trees \cite{chekuri2007multicommodity}.
Without the NBA, however, the natural LP has a super-constant integrality gap. The first theoretical progress for the non-NBA setting was a quasi-PTAS  when the supply graph is a path \cite{bansal2006quasi}. A sequence of papers has ultimately yielded a constant-factor approximation (and integrality gap) for paths, the best of which is an $O(1+\frac1{1+e}+\epsilon)$-approximation \cite{grandoni2020unsplittable}, and an LP with integrality gap $7+\epsilon$ \cite{bonsma2014constant}. For trees, however,
the strongest result is an $O(\log^2 n)$-approximation
 \cite{chekuri2009unsplittable,friggstad2015linear,adamaszek2016submodular}. It
remains an open question whether ANF-Tree has an $O(1)$-approximation (or even an $O(\log n)$-approximation),
and whether ANF-Path has a PTAS.

For trees we use the following notation. An instance $\I=(T,R)$ of ANF-Tree
    consists of an undirected capacitated tree $T=(V,E,u)$ and a set of requests $R$, defined as follows.
    $V$ is the set of vertices and each edge $e\in E$ has some positive capacity $u_e$.
    Each request $r\in R$ imposes some
    non-negative demand $d_r$ on all edges along the unique simple path $P_r$ between $s_r\in V$ and $t_r\in V$. A request may also have a profit $w_r$.
    We assume that $s_r\ne t_r$ for all $r\in R$.
    For each edge $e$, let $R_e=\{r\in R:e\in P_r\}$.
    We denote  $\r=|R|$  and  $m=|E|$.
   A subset
    $S\subseteq R$ of requests is {\em feasible} or {\em routable} if, for each edge $e\in E$, the total demand of all requests $r\in S\cap R_e$
    is at most the capacity $u_e$.
    The goal is to select a feasible subset $S\subseteq R$ which maximizes the profit
    $\sum_{r\in S}w_r$.
    We formalize this with the following IP.
    \[
        \linearprogram{ANF-IP}{\max}{w^Tx}{
            {\sum_{r\in R_e}d_rx_r\le u_e&\forall\,e\in E},
            {x_r\in\{0,1\}^\r&\forall\,r\in R},
        }
    \]
    The natural LP relaxation ANF-LP is defined by replacing $x\in \{0,1\}^\r$ with $x\in[0,1]^\r$;
    ANF-Path is defined similarly.

    One approach for strengthening ANF-LP is to add
    {\em rank constraints} \cite{chekuri2009unsplittable,friggstad2015linear}. For
    $S\subseteq R$ its rank is defined as $\rank(S):=\max\{|T|:T\subseteq S\text{ and }T\text{ is feasible}\}$, and
    the {\em rank constraint}  is then $\sum_{i\in S}x_i\le\rank(S)$.
    Adding all such inequalities to ANF-LP defines the Rank-LP.
    We denote by $P_{\rank}$ the polytope obtained from adding all rank constraints to an ANF-Tree relaxation $P$.
    Rank-LP is NP-Hard to separate, but it can be $O(1)$-approximated \cite{chekuri2009unsplittable,friggstad2015linear}.

   We summarize the known results for these general ANF-Tree formulations.
   \begin{theorem}
   \label{thm:oldresults}~
       \begin{enumerate}
           \item The integrality gap of ANF-LP is $\Omega(\r)$ \cite{chakrabarti2002approximation}.
           \item For ANF-Path, the integrality gap of Rank-LP is $O(\log \r)$
               and the best known lower bound is $\Omega(1)$ \cite{chekuri2009unsplittable}.
            \item For ANF-Tree, the integrality gap of Rank-LP is $O(\log^2 \r)$
               and the best known lower bound is $\Omega(\sqrt{\log \r})$ \cite{friggstad2015linear}.
       \end{enumerate}
   \end{theorem}

\subsection{Knapsack Hierarchy and Strengthening ANF-Tree Relaxations}

In the rest of the paper, $P$ refers to the feasible region of ANF-LP and hence $P_{\rank}$  refers to same for Rank-LP.
The first  result shows  the general dependence of the integrality gap for $P^t$  on $k$ and $t$; this is similar to the Sherali-Adams hierarchy \cite{chekuri2009unsplittable} (although the proofs are not similar).
The upper bound part is  proved in \cref{lem:ktupperbound}, and the lower bound is proved in \cref{lem:ktlowerbound} - see \cref{sec:prelim}

\begin{theorem}
    \label{thm:generalUB}
    For ANF-Tree, the integrality gap of $P^t$ is $O(\r/t)$ and there are instances where it is $\Omega(k/t)$.
\end{theorem}

We now focus on the {\em Friggstad-Gao} instances (definition in \cref{sec:prelim}). These instances established the $\Omega(\sqrt{\log \r})$ ANF-Tree lower bound in \cref{thm:oldresults}; this is significant since it established a super-constant lower bound even when all rank constraints are added. Furthermore, these are {\em single-sink} instances, i.e., all requests share one common endpoint.
We establish the following lower  bounds for the knapsack hierarchy
on the Friggstad-Gao instances.

\begin{theorem}
\label{thm:fglb}
For constant $t$, the integrality gap of $P^t_{\rank}$ is $\Omega(\sqrt{\log \r})$.
For any $c>0$, the integrality gap of both $P^{\r^c}$ and $P^{\r^c}_{\rank}$ is $\Omega(1/c)$.
\end{theorem}

Interestingly, the proof of this lower bound depends  on an upper bound proof. Namely, we define a colouring problem for which we upper bound the chromatic number (see \cref{sec:treelb}).
We can also show that our analysis of $P^t_{\rank}$ on the Friggstad-Gao instances is tight in the following sense.

\begin{theorem}
\label{thm:fgub}
On the class of Friggstad-Gao instances, for any $c>0$, the integrality gap of both $P^{\r^c}$ and $P^{\r^c}_{\rank}$ is $O(1/c)$.
\end{theorem}

It remains an intriguing question whether this  constant integrality gap of $O(1/c)$ holds for general tree instances,  even in the single-sink scenario.

\subsection{Related Work}

    The use of hierarchies for integer programs dates back to the notion of
    Chv\'atal rank \cite{chvatal1973edmonds}. The {\em Chv\'atal closure} of a
    polyhedron $P$ is the polyhedron $P' \subseteq P$ which is defined by the
    system of  Chv\'atal-Gomory cutting planes obtainable from $P$. If we
    denote $P_C^1=P'$ then the hierarchy is generated by
    $P_C^{t+1}=(P_C^t)'$. Chv\'atal proved that $P \supseteq P_C^1 \supseteq
    \ldots \supseteq P_C^{n-1} \supseteq P_C^n = P_I.$
    As discussed, it is NP-hard to separate over $P^1$ both in the Chv\'atal hierachy and the knapsack hierarchy considered in this paper. Other hierarchies have since been
    introduced and widely studied, such as the hierarchy defined by the split closure \cite{cook1990chvatal}
    and hierarchies introduced by
    Lov\'asz-Schrijver \cite{lovasz1991cones}, Sherali-Adams
    \cite{sherali1990hierarchy}, Parillo \cite{parrilo2003semidefinite}, and
    Lasserre \cite{lasserre2001explicit}. These other hierachies also have that the integer hull is obtained after $n$ rounds of the hierarchy, where $n$ is the number of variables. In that sense,
    the knapsack hierarchy is different since $P^{m}=P_I$ where $m$ is the number of constraints. For ANF-tree formulations, however the number of variables equals $\r$, the number of requests.
    Moreover, for ANF-Tree instances one may show that $m\le 4\r$ (see Appendix A.3 in \cite{friggstad2015linear}). Hence
    the knapsack hierarchy is equal to the integer hull at level $O(\r)$ for ANF-Tree.

    We summarize the existing work on the effectiveness of classical hierarchies on ANF-Tree.
    Friggstad and Gao showed that the Lov\'asz-Schrijver hierarchy is ineffective at reducing
    the integrality gap of ANF-Tree after 2 rounds and amounts to adding the rank one
    constraints \cite{friggstad2015linear}.
    Additionally, Chekuri, Ene, and Korula prove that after applying $t$ rounds of the Sherali-Adams hierarchy
    to ANF-LP, the integrality gap is $\Omega(\r/t)$ \cite{chekuri2009unsplittable},
    matching the result for our hierarchy.
    For the case of 0-1 knapsack, Karlin, Mathieu, and Nguyen show that $t^2$ rounds of Lasserre reduce the
    integrality gap to $t/(t-1)$ \cite{karlin2010integrality}. We are not aware of  any work done on whether this would generalize to ANF-Tree.

    In the remainder of the paper we introduce the well-known ``bad gap instances'' for ANF: in particular, the so-called staircase and Friggstad-Gao instances (in \cref{sec:prelim}). We then prove our results (in \cref{sec:first2proofs,sec:khlb,sec:upperbound}), and discuss future work (in \cref{sec:conclusion}).

\section{Preliminary Proofs}

\label{sec:first2proofs}
We begin by showing \cref{prop:approx}, establishing tractability of our hierarchy for
constant $t$.
\begin{proof}[\cref{prop:approx}]
    We use a result by Pritchard  which gives
    a $(1-\e)$-approximate extended formulation for $K_I(S)$ with
    size $O(n^{1+t^3\epsilon^{-1}})$ \cite{pritchard2010lp}.
    For $0<\epsilon\le 1$, denote the projection of this extended formulation onto
    $\bb R^n$ by $K_{\epsilon}(S)$.
    Furthermore, denote by $P^t_{\epsilon}$ the polytope $\bigcap_{|S|=t}K_{\epsilon}(S)$.
    Since $K_{\epsilon}(S)$ is a {\em polyhedral approximation},
    we have $(1-\epsilon)K_{\epsilon}(S)\subseteq K_I(S)\subseteq K_{\epsilon}(S)$ \cite{pritchard2010lp}.
    It follows that
    \begin{align*}
        (1-\e)P^t_{\e}=(1-\e)\bigcap_{|S|=t}K_\e(S)&=
        \bigcap_{|S|=t}(1-\e)K_\e(S)\sube\bigcap_{|S|=t}K_I(S)=P^t
        \intertext{and}
        P^t=\bigcap_{|S|=t}K_I(S)&\sube\bigcap_{|S|=t}K_\e(S)=P^t_\e.
    \end{align*}
    Therefore, $P^t_\e$ is a polyhedral $(1-\e)$-approximate extended formulation for $P^t$.
    There are $\binom n t=O(n^t)$ sets $S$ with $|S|=t$,
    so since each $K_\e(S)$ has size $O(n^{1+t^3\epsilon^{-1}})$, $P^t_\e$
    has size $O(n^{1+t^3\epsilon^{-1}})\cdot O(n^t)=O(n^{t^3\epsilon^{-1}+t+1})$ as desired.\qed
\end{proof}
The proof of the upper bound in \cref{thm:generalUB} uses a strategy which appears again in the proof
of \cref{thm:fgub}: pick some $x\in P^t$ and
partition the requests into sets $S_1,\dots,S_q$ such that the profit of each $x_{S_i}$ (i.e., the vector $x$ but with elements not in $S_i$ set to zero)
can easily be bounded, thus establishing a bound on the profit of $x_{S_1}+\dots+x_{S_q}=x$.
\begin{lemma}
    \label{lem:ktupperbound}
    For ANF-Tree, the integrality gap of $P^t$ is $O(\r/t)$.
\end{lemma}
\begin{proof}
    Let $\OPT$ be the optimal value of $\max\{w^Tx:x\in P_I\}$.
    Consider some set $S\subseteq R$ with $|S|\le t/4$.
    It can be shown that for any ANF-Tree instance
    there exists an equivalent instance (in the sense of having the same integer hull)
    with $m\le4\r$ (see Appendix A.3 in \cite{friggstad2015linear}),
    so we assume w.l.o.g.~that $m\le4\r$.
    So, if the problem is reduced to
    contain only the requests in $S$, then at most $t$ edges
    are needed to define the integer hull.
    Let $T$ be this set of at most $t$ edges.
    Then, we must have $\max\{w^Tx:x\in K_I(T),\,x_{R\setminus S}=0\}\le\OPT$
    because any such $x$ is in $P_I$.

    We can arbitrarily partition the requests into $q:=O(\frac{\r}{4t})$ such sets $S_1,\dots,S_q$ (i.e., with $S_i\subseteq R$ and $|S_i|\le t/4$)
    with corresponding edge sets $T_1,\dots,T_q$ (i.e., where $|T_i|\le t$ and $T_i$ defines the integer hull for requests $S_i$).
    Then $\max\{w^Tx:x\in P^t\}\le\sum_{i=1}^q\max\{w^Tx:x\in K_I(T_i),\,x_{R\setminus  {S_i}}=0\}\le O(\frac{\r}{4t}\OPT)$,
    so the integrality gap is $O(\r/t)$ as desired.\qed
\end{proof}

\section{ANF-Tree Preliminaries}
\label{sec:prelim}

\subsection{Staircase Instances}
\label{sec:staircase}
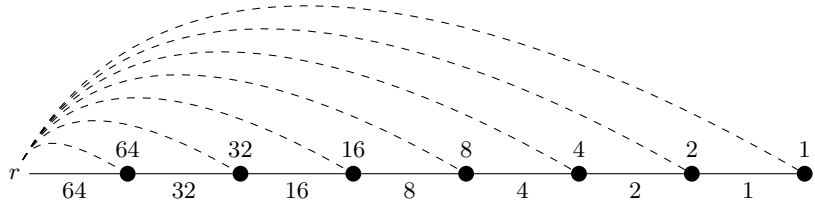
\begin{figure}[t]
    \centering
    \begin{tikzpicture}[
        dot/.style = {circle, fill, minimum size=6pt, inner sep=0pt, outer sep=0pt}
        ]
        \node [] (r) {$r$}
            child[grow=right]{ node [dot,label=above:$64$] (c1) {}
                child{ node [dot,label=above:$32$] (c2) {}
                    child{ node [dot,label=above:$16$] (c3) {}
                        child{ node [dot,label=above:$8$] (c4) {}
                            child{ node [dot,label=above:$4$] (c5) {}
                                child{ node [dot,label=above:$2$] (c6) {}
                                    child{ node [dot,label=above:$1$] (c7) {}
                                    edge from parent node[draw=none,below] {$1$}}
                                edge from parent node[draw=none,below] {$2$}}
                            edge from parent node[draw=none,below] {$4$}}
                        edge from parent node[draw=none,below] {$8$}}
                    edge from parent node[draw=none,below] {$16$}}
                edge from parent node[draw=none,below] {$32$}}
            edge from parent node[draw=none,below] {$64$} };
    \path [draw,dashed] (r) to[out=60,in=150] (c1) node[] {};
    \path [draw,dashed] (r) to[out=60,in=150] (c2) node[] {};
    \path [draw,dashed] (r) to[out=60,in=150] (c3) node[] {};
    \path [draw,dashed] (r) to[out=60,in=150] (c4) node[] {};
    \path [draw,dashed] (r) to[out=60,in=150] (c5) node[] {};
    \path [draw,dashed] (r) to[out=60,in=150] (c6) node[] {};
    \path [draw,dashed] (r) to[out=60,in=150] (c7) node[] {};
    \end{tikzpicture}
    \caption[Instance $S^7$.]{$S^7$ is an ANF-Path instance with $8$ vertices and $7$ requests.
        Each vertex marked with a bullet ($\bullet$) is associated with a request
        (dashed line) which routes between that vertex and $r$.
        The value under each edge denotes the capacity of that edge,
        and the value above each vertex denotes the demand of the request associated
    with that vertex. All requests have profit $1$.}
    \label{fig:badpath}
\end{figure}

For $\r\ge1$, we define the {\em staircase}\footnote{In the literature, this
    instance is referred to as a staircase because of a common way of visualizing ANF-Path instances
where the capacity is plotted above the vertices on the Y axis.}
ANF-Path instance $S^\r=(T,R)$ as follows.
Let $T$ be a path graph on $\r+1$ vertices, that is, $V=\{1,\dots,\r+1\}$ and $E=\{(1,2),(2,3),\dots,(\r,\r+1)\}$.
We refer to vertex $1$ as the root or $r$.
For each $i=1,\dots,\r$, define $u_{(i,i+1)}=2^{i-1}$
and create a request $i$ with $s_i=i$, $t_i=r$, $d_i=2^{i-1}$, and $w_i=1$.
See \cref{fig:badpath} for an illustration.
These instances were first described by Chakrabarti, Chekuri, Gupta, and Kumar \cite{chakrabarti2002approximation}.

\subsection{Friggstad-Gao Instances}
\label{sec:fg}

In this section, we describe the family of Friggstad-Gao ANF-Tree instances, which were introduced in \cite{friggstad2015linear}.

\begin{figure}[t]
    \centering
    \begin{tikzpicture}[
        level 2/.style={sibling distance = 2cm},
        level 3/.style={sibling distance = 0.5cm},
        level/.style={level distance = 1.5cm},
        dot/.style = {circle, fill, minimum size=6pt, inner sep=0pt, outer sep=0pt}
        ]
        \node [] (r) {$r$}
            child{ node [dot] (c3) {}
                child{ node [dot] {}
                    child{ node [dot] {}
                        child [grow=left] {node (q) {$2^{3}-2^{0}$} edge from parent[draw=none]
                            child [grow=up] {node (q) {$2^{6}-2^{3}$} edge from parent[draw=none]
                                child [grow=up] {node (q) {$2^{9}-2^{6}$} edge from parent[draw=none]
                                    child [grow=up] {node (q) {$d_v$} edge from parent[draw=none]
                        }}}
                        child [grow=left] {node (q) {$2^{-4}$} edge from parent[draw=none]
                            child [grow=up] {node (q) {$2^{-2}$} edge from parent[draw=none]
                                child [grow=up] {node (q) {$2^{0}$} edge from parent[draw=none]
                                    child [grow=up] {node (q) {$w_v$} edge from parent[draw=none]
                    }}}}
                        child [grow=right,yshift=0.75cm] {node (q) {$2^{3}$} edge from parent[draw=none]
                                child [grow=up] {node (q) {$2^{6}$} edge from parent[draw=none]
                                    child [grow=up] {node (q) {$2^{9}$} edge from parent[draw=none]
                                        child [grow=up,yshift=-0.8cm] {node (q) {$u_e$} edge from parent[draw=none]
                        }}}}}
                    }
                    child{ node [dot] {}}
                    child{ node [dot] {}}
                    child{ node [dot] {}}
                }
                child{ node [dot] {}
                    child{ node [dot] {}}
                    child{ node [dot] {}}
                    child{ node [dot] {}}
                    child{ node [dot] {}}
                }
                child{ node [dot] {}
                    child{ node [dot] {}}
                    child{ node [dot] {}}
                    child{ node [dot] {}}
                    child{ node [dot] {}}
                }
                child{ node [dot] (c2) {}
                    child{ node [dot] {}}
                    child{ node [dot] {}}
                    child{ node [dot] {}}
                    child{ node [dot] (c1) {}
                        }
                }
            };
    \end{tikzpicture}
    \caption[Instance $T_{FG}^3$.]{ANF-Tree Instance $T_{FG}^3$. Each vertex marked with a bullet ($\bullet$) is associated with a request which terminates at $r$.
    The values on the left indicate the profits/demands/capacities associated with each level.}
    \label{fig:fginstance}
\end{figure}
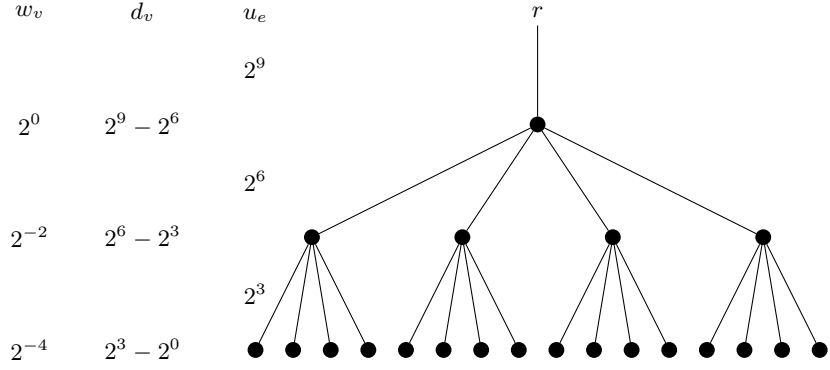

We define the tree $T_{FG}^h$ with height $h\ge2$ as follows. There is a root
vertex $r$ which has a single child $v_1$. Apart from $r$ and the leaves (which are in level $h$), all vertices have
$2^{h-1}$ children. We denote the set of vertices with distance $\ell$ from $r$ by
$\level_\ell$, that is, $\level_0=\{r\}$, $\level_1=\{v_1\}$, and for $\ell\in[h]$,
$|\level_\ell|=2^{(h-1)(\ell-1)}$. For each edge $e=uv$ with $u\in\level_{\ell-1}$ and
$v\in\level_\ell$, define $u_e=2^{h(h-\ell+1)}$. For all $\ell\ge1$ and each
vertex $v\in\level_\ell$, create a request associated with $v$ with $s_v=v$, $t_v=r$, demand
$d_v=2^{h(h-\ell+1)}-2^{h(h-\ell)}$, and profit $w_v=2^{-(h-1)(\ell-1)}$.
See \cref{fig:fginstance} for an example.
This defines a {\em single-sink} instance since every request terminates at $r$.
Moreover, since the profit of each request in any level is the inverse of the number
of requests in that level, the total profit of requests in any level is exactly $1$.
A simple calculation shows that the number of requests (and
equivalently the number of edges) in levels $0$ through $\ell$ is
\begin{equation*}
    \label{eqn:levelnodes}
    n(\ell)=\sum_{i=0}^{\ell-1} (2^{h-1})^i = \frac{2^{(h-1)\ell}-1}{2^{h-1}-1}=\Theta\left(2^{(h-1)(\ell-1)}\right).
\end{equation*}
Thus, $h=\Theta(\sqrt{\log \r})$ where $\r=n(h)$ is the total number of requests/edges,
and for any $\ell$ we have $\ell=\Theta\left(\frac{\log n(\ell)}h\right)$.

The following lemmas establish some fundamental properties of these instances.
We use $T^{<v}$ to denote the requests in the subtree of $T$ rooted at $v$ with $v$ itself removed.

\begin{lemma}
    \label{subtree_lem}
    For any edge $e=uv$ where $u\in\level_{\ell-1}$ and
    $v\in\level_{\ell}$ for some $\ell$, the set of requests in $T^{<v}$ is routable on $e$. That is,
    $1_{T^{<v}} \in K_I(e)$.
\end{lemma}
\begin{proof}
    In any level $\ell'>\ell$,
    $2^{h(h-\ell'+1)}$ is an upper bound for the demand
    $2^{h(h-\ell'+1)}-2^{h(h-\ell')}$ of the requests.  The number of vertices
    in $\level_k$ that are also in the subtree below $e$ is $2^{(h-1)(\ell'-\ell)}$.
    Thus, the demand on $e$ from routing all requests in $\level_{\ell'}$ is at most
    $2^{(h-1)(\ell'-\ell)}2^{h(h-\ell'+1)} = 2^{h^2-h\ell-\ell'+\ell+h}$. Therefore, summing over all
    $\ell'>\ell$, we have
    \begin{align*}
        \textstyle\sum_{\ell'=\ell+1}^h2^{h^2-h\ell-\ell'+\ell+h}
        &=2^{h^2-h\ell+\ell+h}\textstyle\sum_{\ell'=\ell+1}^h2^{-\ell'}\\
        &=2^{h^2-h\ell+\ell+h}(2^{-\ell}-2^{-h})\\
        &\le2^{h(h-\ell+1)}=u_e.\tag*{\qed}
    \end{align*}
\end{proof}

\begin{lemma}
\label{lem:path}
Let $r$ be the root of the tree $T_{FG}^h$ and $P$ be any path from $r$ to a leaf.
Then the demands for the requests of $P$ form a routable set.
\end{lemma}
\begin{proof}
The requests associated with level $\ell$ have demand $2^{h(h-\ell+1)}-2^{h(h-\ell)}$,
so the total demand of such a path is
\begin{align*}
    \sum_{\ell=1}^h 2^{h(h-\ell+1)}-2^{h(h-\ell)}&=2^{h^2}-2^{h(h-1)}
                                        +2^{h(h-1)}-2^{h(h-2)}
                                        +\dots
                                        +2^{2h}-2^{h}
                                        +2^{h}-2^{0}\\
                                        &=2^{h^2}-1.
\end{align*}
This is less than $2^{h^2}$, the capacity of the topmost edge.
A similar argument shows that no other edges are violated by leveraging the self-similar structure of the tree.\qed
\end{proof}

\begin{lemma}
\label{lem:singledge}
    The vector $\frac{1}{2}$ is in $K_I(e)$ for  every edge $e$.
\end{lemma}
\begin{proof}
    Consider any edge $e=uv$ where $u\in\level_{\ell-1}$ and
    $v\in\level_{\ell}$.
    The only requests which route on $e$ are those in the
    subtree rooted at $v$. Therefore it is sufficient
    to show that $b:=\frac{1}{2} (1_{\{v\}} + 1_{T^{<v}}) \in K_I(v)$.
    Note that $1_{\{v\}}$ is in $K_I(e)$
    since \[d_{v}=2^{h(h-\ell+1)}-2^{h(h-l)}<2^{h(h-\ell+1)}=u_e,\]
    and by \autoref{subtree_lem}, we  have that $1_{T^{<v}} \in K_I(e)$. It follows that any convex combination of these vectors, and hence $b$, lies in $K_I(e)$.\qed
\end{proof}

\section{Integrality Gap Lower Bound}
\label{sec:khlb}

In \cref{sec:pathlb}, we
prove a lower bound of $\Omega(k/t)$ on the integrality gap of $P^t$, matching the upper bound shown in \cref{lem:ktupperbound} and thus proving \cref{thm:generalUB}.
However, this lower bound does not hold for $P^t_{\rank}$.
To resolve this case, we show in \cref{sec:treelb}  that on the
Friggstad-Gao tree instances, for any $c>0$ the
integrality gap is reduced to $\Omega(1/c)$ for both $P^{\r^c}$ and $P^{\r^c}_{\rank}$, despite
that $P_{\rank}$ has integrality gap $\Omega(\sqrt{\log \r})$ for these instances.

In the following we assume that all requests are routable on their own, i.e.,
for each $r\in R$ and $e\in P_r$, $d_r\le u_e$.
We also assume that it is impossible to route all requests together,
as the optimal solution would then be trivial.

\subsection{Path Instances}
\label{sec:pathlb}

For path instances, it is known that the integrality gap of $P_{\rank}$ is $O(\log \r)$
and it is conjectured to be $O(1)$ \cite{chekuri2009unsplittable}.
However, the ANF-LP has an integrality gap of $\Omega(\r)$,
which is evidenced by the staircase instances $S^\r$ \cite{chakrabarti2002approximation}.
We now prove the upper bound from \cref{thm:generalUB} by showing that the integrality gap of $P^t$ is $\Omega(\r/t)$.

\begin{lemma}
    \label{lem:ktlowerbound}
    The integrality gap of $P^t$ is $\Omega(\r/t)$.
\end{lemma}
\begin{proof}
    Let $t>1$. We show that $\frac1{t+1}\in P^t$ for instances $S^\r$, as defined in \cref{sec:staircase}.
    Let $S\subseteq E(S^\r)$ with $|S|=t$. For each edge $(i,i+1)\in S$,
    request $i$ is feasible alone. All other requests are feasible together without violating this edge's capacity,
    because any other request $j$ which routes on $(i,i+1)$ has demand $2^j$, edge $(i,i+1)$
    has capacity $2^i$, and $\sum_{j=0}^{i-1}2^j<2^i$.
    These feasible sets define a partition of $R(S^\r)$ into  $t+1$ sets: a set for each of the requests with
    the same indices as the $t$ edges of $S$ and a
    set of all other requests. Since all of these sets are feasible,
    the indicator vector for each of these sets lies in
    $K_I(S)$. Since these sets partition $R(S^\r$), the vector $\frac1{t+1}$  is a convex combination
    of these sets, and hence $\frac1{t+1} \in K_I(S)$.
    Since this holds for every such $S$, we have $\frac1{t+1}\in P^t$ and its total profit
    is $\Omega(\r/t)$, thus establishing the integrality gap.\qed
\end{proof}

\subsection{Tree Instances}
\label{sec:treelb}

In this section, we prove \cref{thm:fglb}, which gives a lower bound on the integrality gap
of $P^t$ on instances $T:=T_{FG}^h$.
Recall that \cref{lem:singledge} establishes $1/2\in P^1$
by proving that for each edge $e$, the $1/2$ vector can be written as a
convex combination of (incidence vectors of) two sets, each of which is
routable on $e$.
We generalize this to any value of $t$ by showing that for $1/c \in P^t$ for
sufficiently small $c$, and thus the integrality gap is
$\Omega(\sqrt{\log \r}/c)$.

Let $S\subseteq E(T)$.
We call a set $X \subseteq R(T)$  {\em $S$-routable} if
$\forall e\in S$, $\sum_{i\in X\cap R(e)} d_i\le u_e$.  Our
key structural result gives a condition when we can express a vector
$1/c$ as a convex combination of $S$-routable sets.

We cast this  convex combination question as a question of colouring the set of
all requests.  For $S \subseteq E(T)$, we define
the {\em $S$-chromatic number}, denoted by $\chi(S)$, to be the minimum value
$c$ such that $R(T)$ can be partitioned into $c$ sets, each of which
is $S$-routable. Given such a partition, the vector $1/c$ is
trivially a convex combination of the indicator vectors of the
$S$-routable sets in the partition.  Thus, if we can show that
$\chi(S) \leq c$ for every $|S|=t$, we have guaranteed that
$1/c \in P^t$.
Hence, the integrality gap established by Friggstad and Gao decreases by at most a
factor of $c/2$  for $P^t$, since the result of Friggstad and Gao is associated with the feasible
vector $1/2 \in P^0$. In fact, the following holds even if we start with the stronger
formulation $P_{\rank}$; we explain why at the end of this section.

\begin{observe}
\label{obs:colour}
The integrality gap of $P^t$ is
$\Omega(h/c)$, where $c(t) :=\max \{ \chi(S): S\sube R, |S|=t \}$.
\end{observe}

\noindent\cref{thm:fglb} follows from the following proposition which the rest of this section is dedicated to proving.

\begin{proposition}
\label{prop:colouring}
If $|S| \leq 2^{h(c-1)}$, then $\chi(S) \leq c+1$.
\end{proposition}
\begin{proof}[\cref{thm:fglb}]
    For constant $t$ and $S\sube R$ with $|S|=t$, $\chi(S)\le2$ for sufficiently large $h$.
    Thus, the integrality gap of $P^t$ is $\Omega(h)=\Omega(\sqrt{\log k})$.

    Now consider some $d>0$ and let $t=k^d$. Let $S\sube R$ with $|S|=t$.
    Then, if $|S| \leq 2^{h(c-1)}$, we have $c=\Omega(\log(k^d)/h)$.
    Hence, $\chi(S)=\Omega(\log(k^d)/h))$, so by Observation~\ref{obs:colour},
    the integrality gap is $\Omega(h^2/\log(k^d))=\Omega(1/d)$.

    To establish that this lower bound holds even when all rank inequalities are added,
we use Theorem 5 from \cite{friggstad2015linear} which
proves that $x/9$ satisfies
all rank constraints if $x$ satisfies all valid constraints of the form $x_i+x_j\le 1$;
these are trivially satisfied by the vector $1/c$ for $c\ge2$.\qed
\end{proof}

Our proof of \cref{prop:colouring} is based on the following colouring result.
The tree $T'$ plays the role of a subtree essentially induced by the edges from some set $S$
with $|S|=t$.
\begin{lemma}
\label{lem:tree}
Let  $T'$ be a subtree of $T$ rooted at some vertex $v$. If  each level of
$T'$ has at most $2^{h(c-1)}$ vertices, then $V(T')$ can be partitioned into
at most $c$ sets which are $E(T')$-routable.
\end{lemma}

\begin{figure}[t]
    \centering
    \begin{tikzpicture}[
        level 1/.style={sibling distance = 5cm},
        level 2/.style={sibling distance = 2cm},
        level 3/.style={sibling distance = 0.5cm},
        level/.style={level distance = 1cm},
        dot/.style = {circle, fill, minimum size=6pt, inner sep=0pt, outer sep=0pt},
        every fit/.style={draw,ultra thick,minimum size=20pt},
        ]
        \node [] (r) {}
            child { node [] {} edge from parent[draw=none] }
            child{ node [] (v) {$v$} edge from parent[dashed]
                child{ node [] (v1) {$v_1$} edge from parent[solid]
                    child{ node [] (l3l) {} edge from parent[dashed] }
                    child{ node []  {$\dots$} edge from parent[dashed]
                        child{ node [] {$\dots$} edge from parent[draw=none]
                            child{ node [] (lcpl) {$\dots$} edge from parent[draw=none]
                                child{ node [] (lcppl) {$\dots$} edge from parent[draw=none]
                                    child{ node [] {$\dots$} edge from parent[draw=none] }
                                }
                            }
                        }
                    }
                    child{ node [] {} edge from parent[dashed] }
                }
                child{ node [] {$\dots$} edge from parent[draw=none]
                    child{ node [] {$\dots$} edge from parent[draw=none]
                        child{ node [] {$\dots$} edge from parent[draw=none]
                            child{ node [] {$\dots$} edge from parent[draw=none]
                                child{ node [] {$\dots$} edge from parent[draw=none]
                                    child{ node [] {$\dots$} edge from parent[draw=none] }
                                }
                            }
                        }
                    }
                }
                child{ node [] (vp) {$v_p$} edge from parent[solid]
                    child{ node [] {} edge from parent[dashed] }
                    child{ node [] {$\dots$} edge from parent[dashed]
                        child{ node [] {$\dots$} edge from parent[draw=none]
                            child{ node [] (lcpr) {$\dots$} edge from parent[draw=none]
                                child{ node [] (lcppr) {$\dots$} edge from parent[draw=none]
                                    child{ node [] {$\dots$} edge from parent[draw=none] }
                                }
                            }
                        }
                    }
                    child{ node [] (l3r) {} edge from parent[dashed] }
                }
            }
            child { node [] {$L_1$} edge from parent[draw=none]
                child { node [] {$L_2$} edge from parent[draw=none]
                    child { node [] {$L_3$} edge from parent[draw=none]
                        child { node [] {$\dots$} edge from parent[draw=none]
                            child { node [] {$L_{c+1}$} edge from parent[draw=none]
                                child { node [] {$L_{c+2}$} edge from parent[draw=none]
                                    child { node [] {$\dots$} edge from parent[draw=none] }
                                }
                            }
                        }
                    }
                }
            }
            ;
        \node[fit=(v),red!50] {};
        \node[fit=(v1) (vp),green!50] {};
        \node[fit=(l3l) (l3r),blue!50] {};
        \node[fit=(lcpl) (lcpr),red!50] (cp1) {};
        \node[fit=(lcppl) (lcppr),green!50] (cp2) {};
    \path [draw,-Triangle,thick] (cp1) to[out=180,in=-180] (v1) node[] {};
    \path [draw,-Triangle,thick] (cp2) to[out=180,in=-180] (l3l) node[] {};
    \end{tikzpicture}
    \caption[A diagram to aid with understanding the proof of \cref{lem:tree}.]
    {A diagram to aid with understanding the proof of \cref{lem:tree}.
    The key observation is illustrated by the arrows on the left: the total demand of every request
    in the box at the tail of an arrow is at most the demand of a single request at the tip of that arrow.
    }
    \label{fig:khlb}
\end{figure}
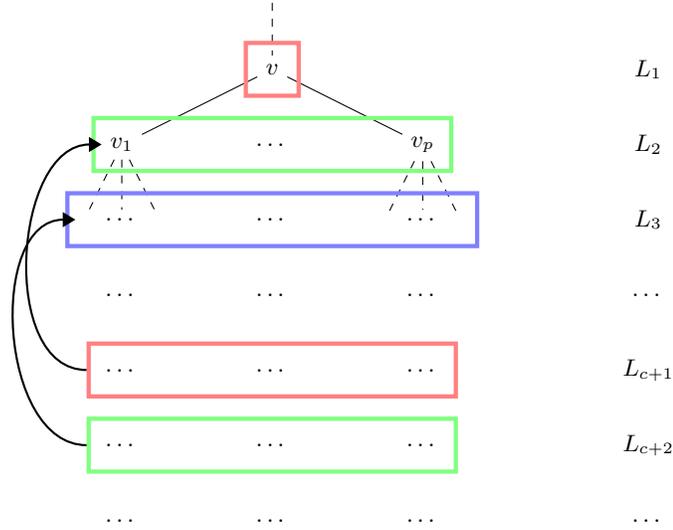
\begin{proof}
We prove this by induction using a stronger induction hypothesis. Specifically, not
only does the colouring exist but we may use the following special type of
colouring. We define layers $L_k$ of $T'$ inductively where $L_1 = \{v\}$. For
each $k \geq 1$, $L_{k+1}$ consists of the children of the requests in layer $L_k$ which
are contained in $T'$. Then for each $i=1,2, \ldots ,c$ we claim that $X_i =
L_i \cup L_{i+c} \cup L_{i+2c} \cup \ldots$ is $E(T')$-routable. Hence,
$X_1,X_2, \ldots ,X_c$ is a valid $c$-colouring which we call {\em layered}. We
claim that a layered colouring always exists for any such subtree $T'$.  The
base case is a single-vertex tree which is trivially true for any $c \geq 1$.

Now consider the children of $v$ in $T'$. Call these $v_1, v_2, \ldots v_p$ and
let $T_i$ be the subtrees of $T'$ associated with each $v_i$.  By induction,
each $T_i$ has a layered colouring which uses at most $c$ colours. Assume we have such a colouring and
without loss of generality that each $v_i$ has colour class $2$,
the next layer below that has colour class $3$, and so on up to colour class $c$,
after which the next layer has colour class $1$.
We show that $v$ can be added to colour class $1$.
Let $X_i$ denote the union of
the colour classes $i$ which occur for the $T_j$. Each layered
colour class $X_i$ is $E(T_i)$-routable and thus is also $E(T')$-routable.
Hence, it only remains to show that $X_1  \cup \{v\}$ is also $E(T')$-routable. Note
that $X_1\cup\{v\}$ consists of layers $L_1 \cup L_{c+1} \cup L_{2c+1} \cup \ldots \cup L_{qc+1}$
of $T'$ for some choice of $q$. Recall that \cref{lem:path} asserts that the
requests along any path from $v$ to the leaves of $T$ is routable.
We show that for all $i$, the total demand of requests of $L_{ic+1}$ is
at most the demand of a single request in $L_{(i-1)c+2}$, so
the total demand from requests in $X_1$ on any edge is at most the demand from routing
a path from $v$ to a leaf, and thus is routable.
See \cref{fig:khlb} for a visual depiction of this.
Suppose this is not the case.
By the self similarity of the tree, we can assume that
the demand of a request in $L_{ic+1}$ is
\[2^{h(h-(ic-1)+1)}-2^{h(h-(ic+1))}=(2^h-1)2^{h(h-ic-1)} \]
and the demand of a request in $L_{(i-1)c+2}$ is
\[ 2^{h(h-((i-1)c+2)+1)}-2^{h(h-((i-1)c+2))}=(2^h-1)2^{h(h-(i-1)c-2)}. \]
Then, we have
\[
|L_{ic+1}| \cdot (2^h-1)2^{h(h-ic-1)} > &~ (2^h-1)2^{h(h-(i-1)c-2)} \\
\Leftrightarrow & &  \\
|L_{ic+1}| > &~\frac{2^{h(h-(i-1)c-2)}}{2^{h(h-ic-1)}} = 2^{h(c-1)},
\]
\noindent
which contradicts our hypothesis.\qed
\end{proof}

\noindent We now complete the proof of \cref{prop:colouring}.

\begin{proof}
Let $v$ be the least common ancestor of the vertices which are incident to the edges
in $S$. We now create a  subtree $T'$ which is a sort of closure of $S$. $T'$ is
obtained by adding edges to $S$ of any path between $v$ and some vertex
incident to an edge $e \in S$. We also include the parent edge
of $v$. We claim that $T'$ satisfies the
hypothesis of \cref{lem:tree}. To see this, consider some level of $T'$
consisting of vertices $a_1, \ldots, a_p$. Let $E_i$ denote the set of edges
which are either incident to  vertex $a_i$ or lie in its subtree.  Note that the
$E_i$ are disjoint. Since each $a_i$ is either incident to an edge of $S$, or
is the internal vertices of some path  used to define the closure $T'$, it follows
that $E_i \cap S \neq \emptyset$ for each $i$, and hence $p \leq \sum_{i=1}^p |E_i \cap S|
\leq |S| \leq 2^{h(c-1)}$.

We now colour all the requests of $T$. We first invoke \cref{lem:tree} to
colour  $R(T')$ using $c$ colours. We
can partition $R(T) \setminus R(T')$ as $A \cup B$, where $B$ denotes the requests
``below'' $T'$ (their paths to the root of $T$ intersect $T'$) and $A$ denotes the
remaining  ``above'' requests. The set $B$ is $S$-routable by \cref{subtree_lem}.
Requests in the set $A$ do not even route on any edge of $S$. Hence, $A \cup B$ can
be the $(c+1)^{st}$ colour class.\qed
\end{proof}

\section{Integrality Gap Upper Bound}
\label{sec:upperbound}

In this section, we prove \cref{thm:fgub},
namely that for instances $T^h_{FG}$ and $c>0$, the integrality gap of both $P^{\r^c}$ and $P^{\r^c}_{\rank}$ is $O(1/c)$.

\begin{theorem}
    \label{thm:newupperbound}
    Let $\ell$ be the largest integer such that $n(\ell)\le t$ (with $n(\ell)$ as defined in \cref{sec:fg}).
    The integrality gap for optimizing over $P^t$ (with profits defined in \cref{sec:fg})
    for instances $T_{FG}^h$ is $O(h/\ell)$.
\end{theorem}

We saw in \cref{sec:fg} that $\ell=\Theta\left({\log(n(\ell))}/h\right)$ and
$h=\Theta\left(\sqrt{\log \r}\right)$. For $c>0$ and $t=\r^c$, the theorem
statement chooses $\ell=\Theta\left({\log(\r^c)}/h\right)$, so the integrality gap is $O(h/\ell)=O(1/c)$,
proving \cref{thm:fgub}.

We show a particular way to partition the requests of the tree into $O(h/\ell)$ sets,
and then show that for each set the profit of any $x\in P^t$ which uses only the requests in that set is $O(1)$.
Since the integral optimum for instances $T_{FG}^h$ is at most $2$ \cite{friggstad2015linear},
it follows that the integrality gap of $P^t$ is $O(h/\ell)$ on these instances.
The proof relies on the self similar structure of the Friggstad-Gao instances, namely that
every vertex except for the leaves and the root has exactly $2^{h-1}$ children
and capacities and demands scale down by $2^h$ for each step away from the root.

For $v\ne r$ let $T^\ell_v$ be the subtree consisting of the first $\ell$ levels of the
children of vertex $v$ along with the edge immediately above $v$. The edge immediately above $v$ has its upper endpoint outside of the subtree.
We denote the edges of the subtree, vertices of the subtree, and requests with an endpoint inside the subtree by $E(T^\ell_v)$, $V(T^\ell_v)$, and $R(T^\ell_v)$, respectively.
For Friggstad-Gao instances, $|E(T^\ell_v)|=|V(T^\ell_v)|=|R(T^\ell_v)|$; we denote this size simply by $|T^\ell_v|$.
Notice that we have $|T_v^\ell|\le n(\ell)$ by self similarity, and this holds with equality unless $v$ is less than $\ell$ levels from the leaves.
Since we assumed $n(\ell)\le t$, we have $|T^\ell_v|\le t$.
For vectors $x\in \bb R^\r$, we denote by $x_{T^\ell_v}$ the restriction of $x$ to those requests with an endpoint in $T^\ell_v$.

We now define, for each $0\le i<\lceil h/\ell\rceil$, a set of subtrees
$\cal P_i=\left\{T_v^\ell:v\in\level_{i\ell+1}\right\}$.
Let $x_{\cal P_i}$ denote the restriction of $x$ to those requests with an endpoint in some $T^\ell_v\in \cal P_i$.
Observe that the union $\cal P=\bigcup\cal P_i$ of these subtrees is a partition of $T^h_{FG}\setminus\{r\}$
into edge and vertex disjoint subtrees. See \cref{fig:khub} for a visual depiction of this.
The following lemma bounds the profit obtainable using requests with an endpoint in some $\cal P_i$.

\begin{figure}[t]
    \centering
    \begin{tikzpicture}[
        level 2/.style={sibling distance = 3cm},
        level 3/.style={sibling distance = 0.65cm},
        level/.style={level distance = 1.5cm},
        dot/.style = {circle, fill, minimum size=6pt, inner sep=0pt, outer sep=0pt},
        hole/.style = {circle, draw=black,fill=white, minimum size=6pt, inner sep=0pt, outer sep=0pt},
        tri/.style={
            draw,dashed,shape border uses incircle,
            isosceles triangle,shape border rotate=90,minimum height=1.5cm,isosceles triangle stretches=true},
        ]
        \node [] (r) {$r$}
            child{ node[dot,yshift=0.25cm] {}
                { node[tri,minimum width=12cm] {}
                    child [yshift=0.75cm] { node [hole] {} edge from parent[draw=none]
                        child [yshift=0.25cm] {  node [dot,yshift=0.5cm] {}
                        { node [tri,minimum width=2.8cm] {} edge from parent[draw=none]
                            child [yshift=0.9cm] {node [hole] {} edge from parent[draw=none]
                                child [yshift=0.25cm] {  node [dot,yshift=0.5cm] {}
                                { node [tri,minimum width=.5cm,minimum height=1.25cm] {} edge from parent[draw=none]} }}
                            child [yshift=0.9cm] {node [hole] {} edge from parent[draw=none]
                                child [yshift=0.25cm] {  node [dot,yshift=0.5cm] {}
                                { node [tri,minimum width=.5cm,minimum height=1.25cm] {} edge from parent[draw=none]} }}
                                child{ node [] {$\dots$} edge from parent[draw=none]}
                            child [yshift=0.9cm] {node [hole] {} edge from parent[draw=none]
                                child [yshift=0.25cm] {  node [dot,yshift=0.5cm] {}
                                { node [tri,minimum width=.5cm,minimum height=1.25cm] {} edge from parent[draw=none]} }}
                            child [yshift=0.9cm] {node [hole] {} edge from parent[draw=none]
                                child [yshift=0.25cm] {  node [dot,yshift=0.5cm] {}
                                { node [tri,minimum width=.5cm,minimum height=1.25cm] {} edge from parent[draw=none]} }}
                    }}}
                    child [yshift=0.75cm] {node [hole] {} edge from parent[draw=none]
                        child [yshift=0.25cm] {  node [dot,yshift=0.5cm] {}
                        { node [tri,minimum width=2.8cm] {} edge from parent[draw=none]} }}
                    child{ node [] {$\dots$} edge from parent[draw=none]
                        child [xshift=3.5cm,yshift=-0.5cm]{ node [] {$\dots$} edge from parent[draw=none]}
                        child [yshift=-0.5cm]{ node [] {$\dots$} edge from parent[draw=none]}
                        child [xshift=-3.5cm,yshift=-0.5cm]{ node [] {$\dots$} edge from parent[draw=none]}
                    }
                    child [yshift=0.75cm] {node [hole] {} edge from parent[draw=none]
                        child [yshift=0.25cm] {  node [dot,yshift=0.5cm] {}
                        { node [tri,minimum width=2.8cm] {} edge from parent[draw=none]} }}
                    child [yshift=0.75cm] {node [hole] {} edge from parent[draw=none]
                        child [yshift=0.25cm] {  node [dot,yshift=0.5cm] {}
                        { node [tri,minimum width=2.8cm] {} edge from parent[draw=none]
                            child [yshift=0.9cm] {node [hole] {} edge from parent[draw=none]
                                child [yshift=0.25cm] {  node [dot,yshift=0.5cm] {}
                                { node [tri,minimum width=.5cm,minimum height=1.25cm] {} edge from parent[draw=none]} }}
                            child [yshift=0.9cm] {node [hole] {} edge from parent[draw=none]
                                child [yshift=0.25cm] {  node [dot,yshift=0.5cm] {}
                                { node [tri,minimum width=.5cm,minimum height=1.25cm] {} edge from parent[draw=none]} }}
                                child{ node [] {$\dots$} edge from parent[draw=none]}
                            child [yshift=0.9cm] {node [hole] {} edge from parent[draw=none]
                                child [yshift=0.25cm] {  node [dot,yshift=0.5cm] {}
                                { node [tri,minimum width=.5cm,minimum height=1.25cm] {} edge from parent[draw=none]} }}
                            child [yshift=0.9cm] {node [hole] {} edge from parent[draw=none]
                                child [yshift=0.25cm] {  node [dot,yshift=0.5cm] {}
                                { node [tri,minimum width=.5cm,minimum height=1.25cm] {} edge from parent[draw=none]} }}
                    }}}
                }
            };
    \end{tikzpicture}
    \caption[A diagram to aid with understanding the proof of \cref{lem:klayerupperbound}.]{
        The partition of $T_{FG}^h$ used to upper bound the integrality gap of the knapsack intersection hierarchy.
        Each vertex marked by $\bullet$ is associated with a subtree $T_v^\ell$,
        indicated here by a dashed triangle. Each triangle spans $\ell$ layers of the tree, i.e.,
        if a vertex marked by $\bullet$ is in level $k$, then the vertex marked by $\circ$ immediately
        below it are in level $k+\ell-1$.
        The set $\cal P_i$ contains the $i^{th}$ level of subtrees. For example,
        $\cal P_0$ contains the single triangle under $r$ and $\cal P_1$ contains all
        the triangles immediately below that.
    }
    \label{fig:khub}
\end{figure}
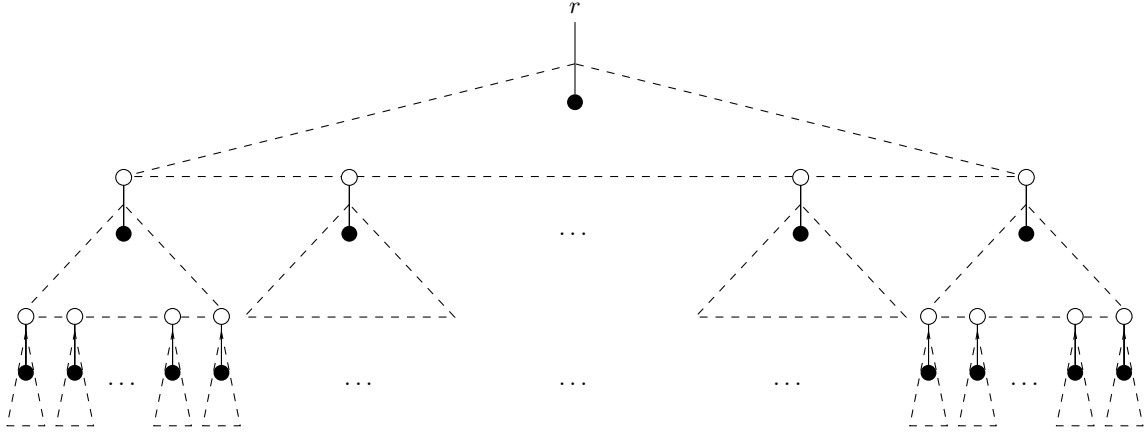
\begin{lemma}\label{lem:klayerupperbound}
    For any feasible vector $x\in P^t$ we have $w^T_{\cal P_i}x_{\cal P_i}\le 2$ for all $0\le i< \lceil h/\ell\rceil$.
\end{lemma}
\begin{proof}
    Let $T_v^\ell\in\cal P_i$. First we show that every feasible subset of $R(T_v^\ell)$ has
    profit at most $2^{-(h-1)i\ell+1}$. This follows by the self similarity of the instance;
    scaling all demands and capacities in $T_v^\ell$ by $2^{hi\ell}$ and all
    profits by $2^{(h-1)i\ell}$ produces a tree identical to $T_{v_1}^\ell$ (recall $v_1$ is the single child vertex of the root $r$).
    For instances $T_{FG}^h$, every routable set has profit at most $2$ \cite{friggstad2015linear}, so
    if we only use requests in $T_{v_1}^\ell$ the profit certainly must be less than $2$.
    By scaling as necessary, it then follows that any feasible subset of $R(T_v^\ell)$ has profit at most $2^{-(h-1)i\ell+1}$,
    as desired.

    Now, we show that to determine feasibility of a subset of $R(T_v^\ell)$ it is
    sufficient to check only the capacity constraints of the edges $E(T_v^\ell)$.
    If $S$ is routable, then clearly no capacity constraints are violated, so assume conversely that $S$ is not routable.
    By \cref{subtree_lem}, no edge which is outside of $E(T_v^\ell)$ and is an ancestor (towards the root) of any
    edge in $S$ has its capacity violated by routing all requests in $T$.
    Furthermore, any other edge which is outside of $E(T_v^\ell)$
    is not routed on by the requests in $S$ and thus cannot be violated.
    Thus, in order for $S$ to not be routable,
    the capacity of one of the edges in $E(T_v^\ell)$ must be violated.

    Since $K_I(E(T_v^\ell))$ is an integer hull, any $x\in K_I(E(T_v^\ell))$ can be written
    as a convex combination of integral vectors in $K_I(E(T_v^\ell))$.
    We saw that to determine feasibility of a subset of $R(T_v^\ell)$ it is sufficient to check the capacity constraints of edges in $E(T_v^\ell)$.
    Thus, for $x\in K_I(E(T_v^\ell))$ such that $x\le1_{R(T_v^\ell)}$,
    we can write $x$ as a convex combination of integral vectors $1_S$ for routable sets $S\subseteq R(T_v^\ell)$, which we know all have profit at most $2^{-(h-1)i\ell+1}$.
    Given $|T_v^\ell|\le t$, any $x\in P^t$ has $x\in K_I(E(T_v^\ell))$, so $w^T_{T_v^\ell}x_{T_v^\ell}\le2^{-(h-1)i\ell+1}$.
    Finally, $|\cal P_i|=|\level_{i\ell+1}|=2^{(h-1)i\ell}$, so we
    can conclude that $w^T_{\cal P_i}x_{\cal P_i}\le2^{-(h-1)i\ell+1}\cdot2^{(h-1)i\ell}= 2$.\qed
\end{proof}
\begin{proof}[\cref{thm:newupperbound}]
    Let $x\in P^t$.
    From \cref{lem:klayerupperbound}, we know that for each $\le i\le\lfloor h/\ell\rfloor$ we have $w^T_{\cal P_i}x_{\cal P_i}\le2$.
    Summing over all $i$, we find that $w^Tx\le2\lfloor h/\ell\rfloor\le 2h/\ell$.
    We know that the integral optimum is $\Omega(1)$, so
    the integrality gap of $P^t$ is $O(h/\ell)$.
    Since the rank formulation is stronger than the natural LP formulation,
    the integrality gap of $P^t_{\rank}$ is $O(h/\ell)$ as well.\qed
\end{proof}

    \section{Conclusion}
    \label{sec:conclusion}

    It would be interesting to establish stronger links to existing
    hierarchies such as those given by Lasserre, Parillo,
    Lov\'asz-Schrijver, Sherali-Adams, or Chv\'atal, or that induced by the split closure. In terms of achieving
    stronger approximations for ANF-Tree, we see two interesting directions.
    One is to consider a rank $t$ approximation $P'^t$ based on intersecting a
    structured set of $t$-row cuts (as opposed to all possible $t$-row cuts, as we have done here).
    This may allow tractable formulations with larger values of $t$.
    A related idea is to consider the intersection of the integer hulls of sub-instances
    induced by keeping a subset of the requests (instead of keeping a subset of the edges).
    For example, to restrict to the set of requests which pass through at least one of some set of $t$ edges,
    as such instances are known to be easier to approximate \cite{grandoni2017augment}.
    Lastly, the question of whether $P_{\rank}^{k^c}$ has constant integrality gap for general
    ANF-Tree instances has so far eluded us; it remains a very interesting question.

~

\noindent
{\bf Acknowledgements:}  The authors are grateful to NSERC for supporting this research.  We would also like to thank Joe Paat for his careful reading and suggestions which improved the presentation.

\bibliographystyle{hep}
\bibliography{references}

\end{document}